\documentclass[12pt]{article}
\usepackage{enumerate}
\usepackage{amsfonts}
\usepackage{latexsym}
\usepackage{color}
\usepackage{graphicx}
\usepackage{wrapfig}
\usepackage{caption}
\usepackage{float}

\usepackage{amsmath,amssymb,amsthm}
\usepackage{bm}
\usepackage{braket}

\def\D{\mathcal{D}}
\def\HH{\mathcal{H}}

\def\St{\mathfrak{S}}
\def\P{\mathfrak{P}}

\def\C{\mathfrak{C}}
\def\T{\mathfrak{T}}
\def\B{\mathfrak{B}}

\newcommand{\supp}{\mathrm{supp}}
\newcommand{\rank}{\mathrm{rank}}

\newcommand{\id}{\mathrm{Id}}
\newcommand{\Tr}{\mathrm{Tr}}
\newcommand{\ext}{\mathrm{ext}}

\newcommand{\shs}{\hspace{1pt}}
\newcounter{defin}  \newcounter{lemma}  \newcounter{theorem}
\newcounter{proposition} \newcounter{corol}  \newcounter{remark} \newcounter{example}
\newenvironment{lemma}{\par\refstepcounter{lemma}     \textbf{Lemma \thelemma.} }{\rm\par}
\newenvironment{theorem}{\par\refstepcounter{theorem}     \textbf{Theorem \thetheorem.}\ }{\rm\par}
\newenvironment{proposition}{\par\refstepcounter{proposition}     \textbf{Proposition \theproposition.}\ }{\rm\par}
\newenvironment{corollary}{\par\refstepcounter{corol}     \textbf{Corollary \thecorol.} }{\rm\par}

\newenvironment{remark}{\par\refstepcounter{remark}     \textbf{Remark \theremark.}}{\rm\par}
\newenvironment{example}{\par\refstepcounter{example}     \textbf{Example \theexample.}}{\rm\par}

\begin{document}

\title{The Moreau-Yosida approximation of the EoF: basic properties and accuracy estimates}




\author{M.E.~Shirokov\footnote{email:msh@mi.ras.ru}\\
Steklov Mathematical Institute, Moscow, Russia}
\date{}
\maketitle
\begin{abstract} We describe a family of convex uniformly continuous functions $E^\lambda_F$, $\lambda>0$, on the set of states of a bipartite quantum system (consisting of finite-dimensional or infinite-dimensional subsystems), which monotonically increase and converge pointwise to the Entanglement of Formation (EoF) as $\lambda\to0$.  These functions are "nonselective" entanglement monotones defined in a way close to the construction of the Moreau-Yosida regularization (the Moreau envelope) of a convex function on a convex set used in the modern convex analysis.  So, we call the functions $E^\lambda_F$ the \emph{Moreau-Yosida approximations} of the EoF and describe their
equivalent definitions and basic properties.

The semicontinuity bounds for the EoF (presented in \cite{LCB+}) allow us to obtain easily computable upper bounds on the difference $E_F(\rho)-E^\lambda_F(\rho)$ for a given state $\rho$.  These bounds give easily computable bounds on the rate of uniform convergence of the function $E^\lambda_F$ to the EoF  as $\lambda\to0^+$ on the sets of states with bounded rank/energy of one of the marginal states.

We also discuss sufficient conditions for the coincidence of  $E_F(\rho)$ and $E^\lambda_F(\rho)$ at a given state $\rho$ for all $\lambda$ small enough and consider several classes of states for which such coincidence takes place.

The conjectured selective LOCC-monotonicity of the functions $E^\lambda_F$ and a possible way to prove it are briefly discussed.

Secondary parts of the article are written with the help of ChatGPT-5.6.
\end{abstract}

\tableofcontents

\section{Introduction}


The Entanglement of Formation (EoF) is one of the basic  measures of
entanglement for bipartite quantum systems \cite{P&V,4H}. For finite-dimensional
quantum systems $A$ and $B$, the EoF of a state $\rho$ of the composite system $AB$
is defined by the convex-roof construction applied to the pure-state
entanglement, which is given by the marginal von Neumann entropy $S(\varrho_A)=S(\varrho_B)$  of a pure
state $\varrho$ of the system $AB$. More explicitly,
\begin{equation}\label{EF-d}
E_F(\rho)=\inf_{\sum_k p_k\varrho_k=\rho}\sum_k p_kS([\varrho_k]_A),
\end{equation}
where the infimum is taken over all finite  ensembles
$\{p_k,\varrho_k\}$ of pure states in $\St(\HH_{AB})$ having $\rho$ as their
average state \cite{Bennett} ($\St(\HH_{AB})$ is the set of all states of the system $AB$).\smallskip

In the case of infinite-dimensional quantum systems $A$ and $B$, there are two ways to generalize the above definition: the discrete and  continuous
convex roof extension of the function
$\,\ext\St(\HH_{AB})\ni\rho\mapsto S(\rho_A)=S(\rho_B)\,$ to the set $\St(\HH_{AB})$, where $\ext\St(\HH_{AB})$ is the set of pure states in $\St(\HH_{AB})$. The former is defined by the formula in (\ref{EF-d}) in which the infimum is taken over all discrete
(finite or countable) ensembles of pure states in $\St(\HH_{AB})$, the latter is defined as
\begin{equation}\label{EF-c}
E_F(\rho)=\!\inf_{\int\varrho\mu(d\varrho)=\rho}\int\! S(\varrho_A)\mu(d\varrho),
\end{equation}
where the  infimum is over all Borel probability measures\footnote{It is reasonable to consider such measures as generalized ensembles of quantum states \cite{H-Sh-2,H-SCI}. From this point of view, a discrete ensemble of pure states is identified with a purely atomic (discrete) probability measure on $\ext\St(\HH_{AB})$.} on the set $\ext\St(\HH_{AB})$ with the barycenter $\rho$ (this infimum is always attained) \cite[Section 5]{EM}.\smallskip

The above definitions coincide if at least one of the marginal entropies of the state $\rho$ is finite \cite[Section 5]{EM}, but the global coincidence of these
definitions is an interesting open problem  (it is clear  that the r.h.s. of (\ref{EF-d}) is not less than the r.h.s. of (\ref{EF-c}) for any state $\rho$, but it is not clear how to prove its coincidence for any mixed states with infinite  marginal entropies).\smallskip

Despite the fact that the discrete definition seems more reasonable from the physical and computational  points of view  there are
serious reasons\footnote{One of the problems with the function defined in (\ref{EF-d}) is the lack of the proof of its vanishing on the set of  countably-non-decomposable separable states  with infinite  marginal entropies \cite{EM}.} to consider the function $E_F$  defined in (\ref{EF-c}) as a true version of the EoF \cite{EM, Lami-new}. Among these
reasons is the global lower semicontinuity of the function $E_F$ defined in (\ref{EF-c}), while the global lower semicontinuity of the discrete version of the EoF is equivalent to its coincidence with the continuous version  defined in (\ref{EF-c}). The lower semicontinuity of $E_F$  makes it possible to guarantee the validity of the basic properties of the entanglement measures (convexity, LOCC-monotonicity, etc.) for this function in the strengthened forms (which are necessary in the infinite-dimensional settings) \cite{EM,Lami-new}. So, in this article \emph{we assume that the EoF is defined by formula  (\ref{EF-c}) keeping in mind that this definition
coincide with (\ref{EF-d}) for any state  $\rho$ of a finite-dimensional bipartite system $AB$ and for any state  $\rho$ of an infinite-dimensional bipartite system $AB$ such that
$\,\min\{S(\rho_A),S(\rho_B)\}<+\infty$}.  \smallskip

An obvious drawback of the definition (\ref{EF-c}) is its technical complexity. It is known that the EoF is a hardly computable function with quite complex behaviour even in the
finite-dimensional case, but the necessity to take  the infimum in (\ref{EF-c}) over all measures on $\ext\St(\HH_{AB})$ with fixed
barycenter radically complicates the definition. Moreover, from the physical point of view, we are faced with the following problem: if $\rho$ is a state with finite
energy (w.r.t. some Hamiltonian), i.e. a  physically realizable state, the supports of some probability measures on $\ext\St(\HH_{AB})$ with the barycenter $\rho$ contain states with infinite
energy which can not be realized in physical experiments.\footnote{This cannot happen with a discrete ensemble of pure states that corresponds to a  probability measure on the set $\ext\St(\HH_{AB})$ with a finite or countable support.}

In this article, we consider a possible way to simplify the analytical and computational problems mentioned before. We describe the properties of the special approximation of the EoF constructed by the way analogous to the definition of the Moreau-Yosida regularization of a convex function widely used in the modern convex analysis \cite{MY-1,MY-2} (it is briefly described at the beginning of Section 3). More specifically, we define the family $\{E^{\lambda}_F\}_{\lambda>0}$ of functions  on the set $\St(\HH_{AB})$ with good analytical properties and more simple and physically relevant definition (compared to (\ref{EF-c})) such that
$$
E^{\lambda}_F(\rho)\nearrow E_F(\rho)\leq +\infty\quad\textrm{ as }\; \lambda\searrow0^+\;\textrm{ for all }\;\rho\in\St(\HH_{AB}).
$$

It is essential that there is a simple way to estimate the accuracy of the approximation of the EoF by the function $E^{\lambda}_F$. The semicontinuity bounds for the EoF (obtained in \cite{LCB+}) can be used to derive easily computable upper bounds on $\,E_F(\rho)-E^{\lambda}_F(\rho)\,$ for a given state $\rho$  depending on simple parameters of $\rho$  (the rank/energy of $\rho_X$, where $X$ is either $A$ or $B$). This allows us to estimate the rate of uniform convergence of the function $E_F^{\lambda}$ to the EoF as $\lambda\to0^+$  on the sets of states with bounded rank/energy of one of the marginal states. In a sense, \emph{the existence of such estimates reduces the problem of calculation of the EoF  with a given accuracy to the problem of
calculation of the function $E_F^{\lambda}$}, which is called the Moreau-Yosida approximation of the  EoF in this article.

\section{Preliminaries}

Throughout the article, we assume that $\mathcal{H}_X$ is a separable Hilbert space describing a quantum system $X$,
$\mathfrak{B}_{\rm sa}(\mathcal{H}_X)$ is the real Banach space of all Hermitian bounded operators on $\mathcal{H}_X$ with the operator norm $\|\cdot\|$ and $\mathfrak{T}(\mathcal{H}_X)$ is the
(complex) Banach space of all trace-class
operators on $\mathcal{H}_X$  with the trace norm $\|\!\cdot\!\|_1$. Write
$\mathfrak{S}(\mathcal{H}_X)$ for   the set of quantum states (positive operators
in $\mathfrak{T}(\mathcal{H}_X)$ with unit trace) \cite{H-SCI,Wilde}. Denote the set of pure states (one-rank projectors) in  $\mathfrak{S}(\mathcal{H}_X)$ by  $\ext\mathfrak{S}(\mathcal{H}_X)$.

Write $I_{X}$ for the unit operator on a Hilbert space
$\mathcal{H}_X$ and  $\id_{X}$ for the identity
transformation of the Banach space $\mathfrak{T}(\mathcal{H}_X)$.

The \emph{von Neumann entropy} of a quantum state
$\rho \in \mathfrak{S}(\HH)$ is  defined by the formula
$S(\rho)=\operatorname{Tr}\eta(\rho)$, where  $\eta(x)=-x\ln x$ if $x>0$
and $\eta(0)=0$. It is a concave lower semicontinuous function on the set~$\mathfrak{S}(\HH)$ taking values in~$[0,+\infty]$ \cite{H-SCI,W,Wilde}.\smallskip

A finite or
countable collection $\{\rho_{i}\}$ of states
with a probability distribution $\{p_{i}\}$ is conventionally called
\textit{discrete ensemble} and denoted by $\{p_{i},\rho_{i}\}$. The state
$\bar{\rho}\doteq\sum_{i}p_{i}\rho_{i}$ is called \emph{average state} of this  ensemble. \smallskip

A \textit{generalized (continuous) ensemble} is defined as
a Borel probability measure on the set of quantum states \cite{H-SCI,H-Sh-2}. We denote by $\P(\mathcal{H})$ the set of all Borel probability measures on $\ext\mathfrak{S}(\mathcal{H})$ equipped with the topology of weak convergence
\cite{Bog,Par}.\footnote{The weak convergence of a sequence $\{\mu_n\}\subset\P(\mathcal{H})$  to a measure $\mu_0\in\P(\mathcal{H})$ means that
$\,\lim_{n\rightarrow\infty}\int f(\rho)\mu_n(d\rho)=\int f(\rho)\mu_0(d\rho)\,$
for any continuous bounded function $f$ on $\,\St(\HH)$.}
 The set $\P(\mathcal{H})$ can be treated as a complete
separable metric space containing the dense subset $\P_{\rm d}(\mathcal{H})$ of discrete measures (corresponding to discrete ensembles) \cite{Par}. The average state of a generalized
ensemble $\mu \in \P(\mathcal{H})$ is the barycenter of the measure
$\mu $ defined by the Bochner integral
\begin{equation}\label{bd}
\bar{\rho}(\mu )=\int_{\ext\mathfrak{S}(\mathcal{H})}\rho \mu (d\rho ).
\end{equation}

Let $A$ and $B$ be quantum systems of any dimensions. We will use the function $S_A$ on the set $\P(\mathcal{H}_{AB})$ defined as
\begin{equation}\label{SA-def}
  S_A(\mu)\doteq\int_{\ext\St(\HH_{AB})}S(\rho_A)\mu(d\rho).
\end{equation}
It is easy to show that $S_A(\mu)$ is an affine lower semicontinuous function on $\P(\mathcal{H}_{AB})$ \cite{H-Sh-2}. Using this function the definition (\ref{EF-c}) can be rewritten as
\begin{equation}\label{EF}
  E_F(\rho)=\inf_{\mu\in\P_\rho(\mathcal{H}_{AB})}S_A(\mu)=\min_{\mu\in\P_\rho(\mathcal{H}_{AB})}S_A(\mu),
\end{equation}
where $\P_\rho(\mathcal{H}_{AB})$ is the subset of $\P(\HH_{AB})$ consisting of measures $\mu$ such that $\,\bar{\rho}(\mu)=\rho$.\medskip

Let $H$ be a positive (semi-definite)  operator on a Hilbert space $\mathcal{H}$ (we will always assume that positive operators are self-adjoint). Write  $\mathcal{D}(H)$ for the domain of $H$. For any positive operator $\rho\in\T(\HH)$ we will define the quantity $\Tr H\rho$ by the rule
\begin{equation}\label{H-fun}
\Tr H\rho=
\left\{\begin{array}{ll}
        \sup_n \Tr P_n H\rho\;\; &\textrm{if}\;\;  \supp\rho\subseteq {\rm cl}(\mathcal{D}(H))\\
        +\infty\;\;&\textrm{otherwise}
        \end{array}\right.
\end{equation}
where $P_n$ is the spectral projector of $H$ corresponding to the interval $[0,n]$ and ${\rm cl}(\mathcal{D}(H))$ is the closure of $\mathcal{D}(H)$. If
$H$ is the Hamiltonian (energy observable) of a quantum system described by the space $\HH$ then
$\Tr H\rho$ is the mean energy of a state $\rho$.

For any positive operator $H$ the set
$$
\C_{H,E}=\left\{\rho\in\St(\HH)\,|\,\Tr H\rho\leq E\right\}
$$
is convex and closed (since the function $\rho\mapsto\Tr H\rho$ is affine and lower semicontinuous). It is nonempty if $E> E_0$, where $E_0$ is the infimum of the spectrum of $H$.

The von Neumann entropy is continuous on the set $\C_{H,E}$ for any $E> E_0$ if and only if the operator $H$ satisfies  the \emph{Gibbs condition}
\begin{equation}\label{H-cond}
  \Tr\, e^{-\beta H}<+\infty\quad\textrm{for all}\;\,\beta>0
\end{equation}
and the supremum of the entropy on this set is attained at the \emph{Gibbs state}
\begin{equation*}
\gamma_H(E)\doteq \frac{e^{-\beta_H(E) H}}{\Tr e^{-\beta_H(E) H}},
\end{equation*}
where the parameter $\beta_H(E)$ is determined by the equation $\Tr H e^{-\beta H}=E\Tr e^{-\beta H}$,
and it is assumed that $\gamma_H(E)|\varphi\rangle=0$ for any $\varphi\in \D(H)^{\perp}$ \cite{W}.

In infinite-dimensions, condition (\ref{H-cond}) can be valid only if $H$ is an unbounded operator having  discrete spectrum of finite multiplicity. It means that
this operator has the spectral representation
  \begin{equation}\label{H-form}
H=\sum_{k=1}^{+\infty} h_k |\tau_k\rangle\langle\tau_k|,
\end{equation}
where $\mathcal{T}\doteq\left\{\tau_k\right\}_{k=1}^{+\infty}$ is the orthonormal
system of eigenvectors of $H$ corresponding to the \emph{nondecreasing} unbounded sequence $\left\{h_k\right\}_{k=1}^{+\infty}$ of its eigenvalues. In this case
$\D(H)=\{ \varphi\in\HH_\mathcal{T}\,| \sum_{k=1}^{+\infty} h^2_k |\langle\tau_k|\varphi\rangle|^2<+\infty\}$, where $\HH_\mathcal{T}$ of the linear span of $\mathcal{T}$.

For any positive operator $H$ satisfying condition (\ref{H-cond}) we will use the function
\begin{equation}\label{F-def}
F_{H}(E)\doteq\sup_{\rho\in\C_{H,E}}S(\rho)=S(\gamma_H(E))=\beta_H(E)E+\ln Z_H(E)
\end{equation}
on $[h_1,+\infty)$, where $Z_H(E)=\Tr e^{-\beta_H(E)H} $  and $h_1$ is the minimal eigenvalue of $H$.

By Proposition 1 in \cite{EC} the Gibbs condition (\ref{H-cond}) is equivalent to the following asymptotic property
\begin{equation}\label{H-cond-a}
  F_{H}(E)=o\shs(E)\quad\textrm{as}\quad E\rightarrow+\infty.
\end{equation}

For example, if $\,N\doteq a^\dag a\,$ is the number operator of a quantum oscillator then $F_{\hat{N}}(E)=g(E)$, where
\begin{equation}\label{g-def}
 g(x)=(x+1)\ln(x+1)-x\ln x,\;\, x>0,\qquad g(0)=0.
\end{equation}

We will often assume that
\begin{equation}\label{star}
  h_1=\inf\limits_{\|\varphi\|=1}\langle\varphi\vert H\vert\varphi\rangle=0.
\end{equation}

For a nonnegative function $f$ on $[0,1]$ we denote  its non-decreasing
envelope by $f^{\uparrow}$, i.e.
\begin{equation*}
  f^{\uparrow}(x)\doteq \sup_{t\in[0,x]} f(t),\quad x\in[0,1].
\end{equation*}
In particular, we will  use  the \emph{nondecreasing concave} continuous function
\begin{equation}\label{h+}
h^{\uparrow}(x)=\left\{\begin{array}{ll}
        h(x) &\textrm{if}\;\;  x\in[0,\frac{1}{2}]\\
        \ln2\quad& \textrm{if}\;\;  x\in(\frac{1}{2},1].
        \end{array}\right.
\end{equation}
where $\,h(x)=\eta(x)+\eta(1-x)\,$ is the binary entropy.

\section{The Moreau-Yosida approximation of the EoF}

\subsection{Definition and basic properties}

The Moreau-Yosida regularization (the Moreau envelope)
$M_{f}$ of a proper lower semi-continuous convex function
$f$ on a convex subset of a metric space is a smoothed version of $f$.  It was proposed by J.J.Moreau in 1965 \cite{MY-1} and widely used in optimization and variational analysis to deal with non-smooth functions.

The Moreau-Yosida regularization of order $\lambda>0$ of a proper convex function
$f$ from a Hilbert space $\HH$ (of any dimension) to $(-\infty,+\infty]$  is defined as (cf. \cite{MY-2})
$$
M^\lambda_{f}(\phi)=\inf _{\psi\in\HH}\left\{f(\psi)+\frac {1}{2\lambda}\|\psi-\phi\|^{2}\right\},\quad \phi\in\HH,
$$
where $\|\cdot\|$ is the norm in $\HH$. The function $M^\lambda_{f}$ has good analytical properties (convexity, Lipschitz continuity, etc.).
Under a certain condition (the lower semicontinuity of $f$) the  Moreau-Yosida regularization $M^\lambda_{f}$ pointwise converges to the function $f$. So, it can be treated
as a smooth approximation of $f$.\medskip

In this article we apply a slightly modified version of the above construction to the EoF (considered as a function on the set of bipartite quantum states).

Let $A$ and $B$ be quantum systems of any dimensions. For a given $\lambda>0$ consider the function
\begin{equation}\label{EFA}
\! E_F^\lambda(\rho)=\!\inf_{\sigma\in\St(\HH_{AB})}\!\left\{E_F(\sigma)+\frac{1}{2\lambda}\|\rho-\sigma\|_1\right\}
 =\!\inf_{\mu\in\P(\HH_{AB})}\!\left\{S_A(\mu)+\frac{1}{2\lambda}\|\rho-\bar{\rho}(\mu)\|_1\right\}\!
\end{equation}
on the set $\St(\HH_{AB})$, where $S_A$ is the function defined in (\ref{SA-def}) and $\bar{\rho}(\mu)$ is the barycenter of $\mu$ defined in (\ref{bd}).\footnote{$\bar{\rho}(\mu)$ can be treated as the average state of a generalized ensemble $\mu$ \cite{H-Sh-2}.} The second equality here is due the possibility to permute the infima (i.e. the property $\inf_{x\in X}\inf_{y\in Y}f(x,y)=\inf_{y\in Y}\inf_{x\in X}f(x,y)$).\smallskip

Since the definition of the function $E_F^\lambda$ is very close\footnote{It is essential that we use  $\|\rho-\sigma\|_1$ instead of  $\|\rho-\sigma\|^2_1$ in (\ref{EFA}).} to the definition of the  Moreau-Yosida regularization of a
convex function described before, we will call it the \emph{Moreau-Yosida approximation} (briefly, the \emph{$MY$-approximation}) of the EoF of order $\lambda$.\smallskip

In the following theorem we describe basic properties of the function $E_F^\lambda$.\smallskip

\begin{theorem}\label{P-1} \emph{Let $A$ and $B$ be quantum systems of any dimensions.}\smallskip

\noindent A) \emph{For  given $\lambda>0$ and a state $\rho$ in $\St(\HH_{AB})$ the following inequalities hold
$$
E_F^\lambda(\rho)\leq\frac{{\rm ed}(\rho)}{\lambda}\leq\frac{1}{\lambda},
$$
where $\,{\rm ed}(\rho)\doteq\frac{1}{2}\inf\{\|\rho-\sigma\|_1\,|\,\sigma\in\St_{\rm sep}(\HH_{AB})\}\,$ is the distance from the state $\rho$ to the set $\,\St_{\rm sep}(\HH_{AB})$
of separable states in $\,\St(\HH_{AB})$.}\smallskip

\noindent B) \emph{For a given $\lambda>0$ the function
$E_F^\lambda$ is convex and Lipschitz continuous on $\St(\HH_{AB})$ with the Lipschitz constant not exceeding $\frac{1}{\lambda}$:}\footnote{Following
 the convention we assume that the metric on $\St(\HH_{AB})$ is defined as $\,\frac{1}{2}\|\rho-\sigma\|_1$ \cite{H-SCI,Wilde}.}
$$
|E_F^\lambda(\rho)-E_F^\lambda(\sigma)|\leq\frac{1}{2\lambda}\|\rho-\sigma\|_1\qquad \forall \rho,\sigma\in\St(\HH_{AB}).
$$

\noindent C) \emph{For an arbitrary  state $\rho\in\St(\HH_{AB})$ the function  $\lambda \mapsto E_F^\lambda(\rho)$ is non-increasing on $(0,+\infty)$ and}
$$
E_F^\lambda(\rho)\nearrow E_F(\rho)\leq +\infty \qquad\textit{ as }\quad \lambda\searrow 0^+.
$$

\noindent D) \emph{The equality $\,E_F^\lambda(\rho)=E_F(\rho)\,$ holds for a given state $\rho\in\St(\HH_{AB})$  and  $\lambda=\lambda_\rho>0$ (and hence for all $\,\lambda\in (0,\lambda_\rho]\,$) if and only if}\footnote{Criteria for the
validity of (\ref{lsc}) are presented in Theorem 1 in \cite{H&Sh-5}.}
\begin{equation}\label{lsc}
E_F(\rho)-E_F(\sigma)\leq \frac{1}{2\lambda_\rho}\|\rho-\sigma\|_1\quad \forall\sigma\in\St(\mathcal{H}_{AB}).
\end{equation}

\noindent E) \emph{For every $\,\lambda>0$ the  equality  $\,E^\lambda_F(\rho)=0$  holds if and only if $\rho$ is a separable (non-entangled) state in $\St(\HH_{AB})$}.
\newpage

\noindent F) \emph{Let $\,\rho$ be a given state in $\St(\HH_{AB})$. If  $\,\Phi$ is a  transformation of $\,\St(\HH_{AB})$ such that
$$
E_F(\Phi(\sigma))\leq E_F(\sigma)\quad\textit{ and }\quad \|\Phi(\rho)-\Phi(\sigma)\|_1\leq \|\rho-\sigma\|_1
$$
for all $\sigma$ in $\St(\HH_{AB})$ then
$\,E^\lambda_F(\Phi(\rho))\leq E^\lambda_F(\rho)\,$
for all  $\lambda>0$.}\medskip

\noindent  G) \emph{Let $\,\rho$ be a given state in $\St(\HH_{AB})$. If $\,\Phi$ is a  transformation of $\,\St(\HH_{AB})$ such that
$$
\Phi(\rho)=\rho,\quad E_F(\Phi(\sigma))\leq E_F(\sigma)\quad\textit{ and }\quad \|\Phi(\rho)-\Phi(\sigma)\|_1\leq \|\rho-\sigma\|_1
$$
for all $\sigma$ in $\St(\HH_{AB})$ then for any $\lambda>0$ the first (resp. the second) infimum in (\ref{EFA}) can be taken over the set}
$$
\Phi(\St(\HH_{AB}))\doteq \left\{\Phi(\varsigma)\,|\, \varsigma\in\St(\HH_{AB})\right\}\;\; \textit{(resp. }\P_{\Phi}\doteq\left\{\mu\in\P(\HH_{AB})\,|\, \bar{\rho}(\mu)\in\Phi(\St(\HH_{AB}))\right\}\textit{).}
$$

\noindent H) \emph{For an arbitrary  state $\rho\in\St(\HH_{AB})$ and  $\lambda>0$ the second infimum in (\ref{EFA}) can be taken over all discrete (purely atomic) measures in $\,\P(\HH_{AB})$, i.e.
\begin{equation}\label{EFA+}
 E_F^\lambda(\rho)=\inf_{\{p_i,\rho_i\}\in\P_{\rm d}(\HH_{AB})}\left(\sum_i p_i S([\rho_i]_A)+\frac{1}{2\lambda}\|\rho-\bar{\rho}(\{p_i,\rho_i\})\|_1\right),
\end{equation}
where $\,\bar{\rho}(\{p_i,\rho_i\})\doteq\sum_i p_i \rho_i\,$ is the average state of an ensemble $\{p_i,\rho_i\}$ and  $\P_{\rm d}(\HH_{AB})$ is the set of all discrete ensembles of pure states in $\St(\HH_{AB})$. Moreover, if  $\rho$ is a state such that $\Tr H\rho_A<+\infty$ for some positive operator $H$ on $\HH_{A}$ of the form (\ref{H-form})
then the infimum in (\ref{EFA+}) can be taken over all ensembles $\{p_i,\rho_i\}$ in $\P_{\rm d}(\HH_{AB})$ such that $\Tr H[\rho_i]_A<+\infty$ for all $i$.}
\end{theorem}\medskip

\emph{Proof.} A) This claim directly follows from the definition of the function $E^\lambda_F$ because $E_F(\sigma)=0$ for any separable state $\sigma$
in $\St(\HH_{AB})$. \smallskip

B) To prove the convexity of the function $E_F^\lambda$ assume that $\rho_1$  and $\rho_2$ are arbitrary states
in $\St(\HH_{AB})$. Then for any $\varepsilon>0$ there exist  states $\sigma_1$  and $\sigma_2$
in $\St(\HH_{AB})$ such that
$$
E_F^\lambda(\rho_i)\geq E_F(\sigma_i)+\frac{1}{2\lambda}\|\rho_i-\sigma_i\|_1-\varepsilon,\quad i=1,2.
$$
Let $\rho_*=\frac{1}{2}(\rho_1+\rho_2)$ and $\sigma_*=\frac{1}{2}(\sigma_1+\sigma_2)$. Then
$$
\begin{array}{c}
E_F^\lambda(\rho_*)\leq E_F(\sigma_*)+\frac{1}{2\lambda}\|\rho_*-\sigma_*\|_1 \leq  \frac{1}{2}(E_F(\sigma_1)+E_F(\sigma_2))+\frac{1}{2\lambda}\|\frac{\rho_1-\sigma_1}{2}+\frac{\rho_2-\sigma_2}{2}\|_1
\\\\\leq  \frac{1}{2}(E_F(\sigma_1)+\frac{1}{2\lambda}\|\rho_1-\sigma_1\|_1)+\frac{1}{2}(E_F(\sigma_2)+\frac{1}{2\lambda}\|\rho_2-\sigma_2\|_1)\leq \frac{1}{2}(E_F^\lambda(\rho_1)+E_F^\lambda(\rho_2))+\varepsilon,
\end{array}
$$
where the convexity of the function $E_F$ and the triangle inequality for the trace norm were used. Since $\varepsilon$ is arbitrary we obtain $\,E_F^\lambda(\rho_*)\leq\frac{1}{2}(E_F^\lambda(\rho_1)+E_F^\lambda(\rho_2))$.

To prove the Lipschitz continuity  of the function $E_F^\lambda$ assume that $\rho_1$  and $\rho_2$ are arbitrary states
in $\St(\HH_{AB})$. Let $\varepsilon>0$  and  $\sigma_2$ be a state
in $\St(\HH_{AB})$ such that
$$
E_F^\lambda(\rho_2)\geq E_F(\sigma_2)+\frac{1}{2\lambda}\|\rho_2-\sigma_2\|_1-\varepsilon.
$$
Then we have
$$
\begin{array}{c}
 E^\lambda_F(\rho_1)\leq E_F(\sigma_2)+\frac{1}{2\lambda}\|\rho_1-\sigma_2\|_1\leq  E_F(\sigma_2)+\frac{1}{2\lambda}\|\rho_2-\sigma_2\|_1+\frac{1}{2\lambda}\|\rho_2-\rho_1\|_1
\\\\\leq E^\lambda_F(\rho_2)+\frac{1}{2\lambda}\|\rho_2-\rho_1\|_1+\varepsilon,
\end{array}
$$
where the definition of the function $E^\lambda_F$ and the triangle inequality for the trace norm were used. Since $\varepsilon$ is arbitrary, we obtain $\,E_F^\lambda(\rho_1)-E_F^\lambda(\rho_2)\leq\frac{1}{2\lambda}\|\rho_1-\rho_2\|_1$.\smallskip

By permuting the roles of $\rho_1$  and $\rho_2$ we show that  $E_F^\lambda(\rho_2)-E_F^\lambda(\rho_1)\leq\frac{1}{2\lambda}\|\rho_2-\rho_1\|_1$.\smallskip

C) Let $\rho$ be an arbitrary state in $\St(\HH_{AB})$. It is clear that $\,E_F^{\lambda_1}(\rho)\leq E_F^{\lambda_2}(\rho)\,$ for any $\lambda_1>\lambda_2$. It is also clear that
$\,\sup_{\lambda>0}E_F^{\lambda}(\rho)\leq E_F(\rho).$ Hence,
$$
E_F^\lambda(\rho)\nearrow C\leq E_F(\rho)\leq +\infty \quad\textup{as}\quad \lambda\searrow 0^+.
$$
Assume that $\,C<E_F(\rho)$. This implies that $C<+\infty$. Let $\{\lambda_n\}$ be any sequence of positive numbers converging
to zero and $\{\sigma_n\}$ be a sequence of states such that
$$
E_F^{\lambda_n}(\rho)\geq E_F(\sigma_n)+\frac{1}{2\lambda_n}\|\rho-\sigma_n\|_1-\frac{1}{n}\quad \forall n.
$$
Since $C<+\infty$, the sequence $\{\sigma_n\}$ converges to the state $\rho$. So, the lower semicontinuity of $E_F$
implies that
$$
\liminf_{n\to+\infty} E_F^{\lambda_n}(\rho)\geq \liminf_{n\to+\infty}\left(E_F(\sigma_n)+\frac{1}{2\lambda_n}\|\rho-\sigma_n\|_1-\frac{1}{n}\right)\geq E_F(\rho).
$$
This contradicts our assumption.\smallskip

D) This claim directly follows from the definitions of the functions $E_F$ and $E_F^\lambda$.\smallskip

E) Since $E_F^{\lambda}(\rho)\leq E_F(\rho)$ for any state $\rho$, to prove this claim it suffices to show that
$$
\exists\lambda>0\;\textrm{ such that }\; E_F^{\lambda}(\rho)=0\quad \Rightarrow \quad \rho\;\textrm{ is a separable state.}
$$
Assume that $E_F^{\lambda}(\rho)=0$. Then  there is a sequence $\{\sigma_n\}$ of states in $\St(\HH_{AB})$ such that
$$
\lim_{n\to+\infty} \left(E_F(\sigma_n)+\frac{1}{2\lambda}\|\rho-\sigma_n\|_1\right)=0.
$$
Thus, the sequence $\{\sigma_n\}$ converges to the state $\rho$ and
the sequence $E_F(\sigma_n)$ tends to zero. Hence
$$
E_F(\rho)\leq \liminf_{n\to+\infty} E_F(\sigma_n)=0
$$
by the lower semicontinuity of the function $E_F$. So, the state $\rho$ is separable because $E_F$ is an entanglement measure.\smallskip

F) The definition of the function $E_F^\lambda$ implies
 \begin{equation*}
\begin{array}{rl}
E_F^\lambda(\Phi(\rho))\,=&\!\displaystyle\inf_{\sigma\in\St(\HH_{AB})}\!\left\{E_F(\sigma)+\frac{1}{2\lambda}\|\Phi(\rho)-\sigma\|_1\right\}\\
\leq&\!\displaystyle\inf_{\sigma\in\St(\HH_{AB})}\!\left\{E_F(\Phi(\sigma))+\frac{1}{2\lambda}\|\Phi(\rho)-\Phi(\sigma)\|_1\right\}\\
\leq&\!\displaystyle\inf_{\sigma\in\St(\HH_{AB})}\!\left\{E_F(\sigma)+\frac{1}{2\lambda}\|\rho-\sigma\|_1\right\}=E_F^\lambda(\rho).
\end{array}
\end{equation*}

G) This claim follows from the definition of $E_F^\lambda(\rho)$ and the inequality
$$
E_F(\Phi(\sigma))+\frac{1}{2\lambda}\|\rho-\Phi(\sigma)\|_1\leq E_F(\sigma)+\frac{1}{2\lambda}\|\rho-\sigma\|_1\qquad \forall \sigma\in\St(\HH_{AB}) .
$$

H) Assume first that $\rank\rho_A<+\infty$. Theorem \ref{P-1}G implies (see Corollary \ref{Edc} below) that the second infimum
in (\ref{EFA}) can be taken over the set $\P(\supp\rho_A\otimes\mathcal{H}_B)\,$. By Proposition 3 in \cite{AOE} the function $S_A$ is continuous on  $\P(\supp\rho_A\otimes\mathcal{H}_B)$. So,
in this case the first statement of claim H is valid because discrete measures form a dense subset in $\P(\supp\rho_A\otimes\mathcal{H}_B)$ \cite{Par}. The second statement is derived from the first one by replacing $H$ with the operator $PHP$, where $P$ is the projector onto $\supp\rho_A$.

If $\rank\rho_A=+\infty$ then it suffices to prove the second  statement of claim H (as for every $\rho$ one can find a positive operator $H$ of the form (\ref{H-form}) such that
$\Tr H\rho_A<+\infty$).

Assume that the operator $H$ has representation (\ref{H-form}). Let
$\mathcal{H}_A^n$ be the subspace of $\mathcal{H}_A$ generated by the
vectors $\tau_1,\tau_2,\ldots,\tau_n$. Since the condition $\,\Tr H\rho_A<+\infty\,$
implies $\,\supp\rho\subseteq\overline{\bigcup_n\mathcal{H}_A^n\otimes\HH_B},$ this claim can be proved by using Lemma \ref{hl}  below and by noting that for every natural $n$
\begin{itemize}
  \item  the function $S_A$ is continuous on $\,\P(\mathcal{H}_A^n\otimes\HH_B)$ by Proposition 3 in \cite{AOE},
  \item  discrete measures form a dense subset in $\,\P(\mathcal{H}_A^n\otimes\HH_B)$ \cite{Par},
  \item  the function $\,\varsigma\to \Tr H\varsigma\,$  takes finite values on the set $\St(\mathcal{H}_A^n)$.
\end{itemize}
$\Box$\medskip

\begin{lemma}\label{hl}
\emph{Let $\,\mathfrak{H}\doteq\{\mathcal{H}_A^n\}_{n\in\mathbb{N}}$ be a family of subspaces
of $\,\mathcal{H}_A$ such that $\,\mathcal{H}_A^n\subseteq\mathcal{H}_A^{n+1}$
for all $n$. Then for any state $\,\rho\in\St(\mathcal{H}_{AB})$ such that $\,\supp\rho\subseteq\overline{\bigcup_{n\in\mathbb{N}}\mathcal{H}_A^n\otimes\HH_B}\,$ the infimum
in the second expression in (\ref{EFA}) can be taken over the subset}
\[
\P_{\mathfrak{H}}
=
\left\{
\mu\in\P(\mathcal{H}_{AB})
\;\middle|\;
\exists n:\
\operatorname{supp}\mu
\subset
\mathcal{H}_A^n\otimes\mathcal{H}_B
\right\},
\]
\end{lemma}

\medskip

\noindent
\emph{Proof.} Let $P_n$ be the projector onto
$\mathcal{H}_A^n\otimes\mathcal{H}_B$ and $\,\rho_n
=
\frac{P_n\rho P_n}{\operatorname{Tr}P_n\rho}$ be a state in $\St(\HH_{AB})$  (we assume that $n$ is so large that $\operatorname{Tr}P_n\rho>0$).
Then
\[
\varepsilon_n
\doteq
\frac{1}{2\lambda}
\left\|
\rho-\rho_n
\right\|_1
\]
tends to zero as $n\to\infty$. Since
$\,\operatorname{supp}\rho_n
\subset
\mathcal{H}_A^n\otimes\mathcal{H}_B$, Theorem \ref{P-1}G implies that there is a measure
$\,\mu_n\in
\P(\mathcal{H}_A^n\otimes\mathcal{H}_B)\,$
such that
\[
E_F^\lambda(\rho_n)
\geq
S_A(\mu_n)
+
\frac{1}{2\lambda}
\left\|
\rho_n-\overline{\rho}(\mu_n)
\right\|_1
-
\frac{1}{n}.
\]
At the same time, Theorem \ref{P-1}B  implies that $\,E_F^\lambda(\rho)\geq E_F^\lambda(\rho_n)-\varepsilon_n$. So, we obtain
$$
E_F^\lambda(\rho)\geq S_A(\mu_n)
+
\frac{1}{2\lambda}
\left\|
\rho_n-\overline{\rho}(\mu_n)
\right\|_1
-
\frac{1}{n}-\varepsilon_n\geq S_A(\mu_n)
+
\frac{1}{2\lambda}
\left\|
\rho-\overline{\rho}(\mu_n)
\right\|_1
-
\frac{1}{n}-2\varepsilon_n.
$$
Since $\,\frac{1}{n}+2\varepsilon_n\,$ tends to zero as $\,n\to+\infty$, this implies the claim of the lemma. $\Box$\medskip

Theorem \ref{P-1}F and the LOCC-monotonicity of the EoF imply the following\smallskip

\begin{corollary}\label{E-r} \emph{For every $\lambda>0$ the function $E^\lambda_F$ does not increase under nonselective
LOCC-operations (cf.\cite{Vidal,P&V,Lami-new}).}
\end{corollary}
\medskip

Due to Corollary \ref{E-r} and claims $B$ and $E$ of Theorem \ref{P-1} we may say that  for every $\lambda>0$ \emph{the function $E^\lambda_F$
is a "nonselective" entanglement monotone}.\footnote{To show that $E^\lambda_F$ is a true entanglement monotone (in terms of \cite{Vidal,P&V}) it suffices to prove that $E^\lambda_F$ does not increase under selective LOCC-operations. This is an open question discussed in  Section 7 below.}
\medskip

Theorem \ref{P-1}G allows us to prove the following\smallskip
\begin{corollary}\label{Edc}\emph{ For any $\lambda>0$ the first (resp. the second) infimum in (\ref{EFA}) can be taken over the set
$\,\St(\HH^\rho_A\otimes\HH^\rho_B)\,$ (resp. $\,\P(\HH^\rho_A\otimes\HH^\rho_B)$), where $\,\HH^\rho_X\doteq\supp\rho_X$, $X=A,B$.}
\end{corollary}\smallskip

To derive Corollary \ref{Edc} from Theorem \ref{P-1}G one should take the channel
$$
\Phi(\sigma)=P_A\otimes P_B\cdot \sigma\cdot P_A\otimes P_B + [\Tr(I_{AB}-P_A\otimes P_B)\sigma]\shs \tau,\quad \sigma\in\St(\HH_{AB}),
$$
where $P_X$ is the projector onto $\HH^\rho_X$, $X=A,B$, and $\tau$ is an arbitrary separable state in $\,\St(\HH^\rho_A\otimes\HH^\rho_B)$.
\medskip

\begin{remark}\label{Edr}
Corollary \ref{Edc} shows the invariance of the definition of the function $E^\lambda_F$  w.r.t. to the embeddings $\,\HH_A\subset\HH_{A'}$ and $\,\HH_B\subset\HH_{B'}$:
$$
\inf_{\sigma\in\St(\HH_{AB})}\!\left\{E_F(\sigma)+\frac{1}{2\lambda}\|\rho-\sigma\|_1\right\}=\inf_{\sigma\in\St(\HH_{A'B'})}\!\left\{E_F(\sigma)+\frac{1}{2\lambda}\|\rho-\sigma\|_1\right\}\quad \forall\lambda>0.
$$
\end{remark}\medskip

\begin{remark}\label{r2}
Theorem \ref{P-1}H  implies, in particular, that in the case of infinite-dimensional quantum systems $A$ and $B$ the definition
of $E^\lambda_F(\rho)$ for any state $\rho$ with finite marginal energy of one of the subsystems (w.r.t. some Hamiltonian $H$ in this subsystem)  requires optimization only over
discrete ensembles of pure state each of which has finite marginal energy w.r.t. $H$. So, we may say that\emph{ the definition
of the function $E^\lambda_F$ is physically relevant for all states of the system $AB$ produced in real experiments.}
\end{remark}

\subsection{Variational expression for the function $E^{\lambda}_F$}

It is well known (cf.\cite{Ulm,H-SCI,H&Sh-5}) that\footnote{Here and in what follows we denote the state $|\psi\rangle\langle\psi|$ by $\psi$ for brevity.}
\begin{equation}\label{EV}
   E_F(\rho)
    =
    \sup_{\Lambda\in\mathcal{B}_{\mathrm{sa}}(\mathcal{H})}
    \inf_{\psi\in\mathcal{H}^1_{AB}}
    \left\{
        S(\psi_A)
        +
        \operatorname{Tr}\Lambda
        \bigl(\rho-\psi\bigr)
    \right\}=\sup_{\Lambda\in\mathfrak{D}\!\left(+\infty\right)}\Tr \Lambda\rho,
\end{equation}
where $\HH^1_{AB}$ is the unit sphere in $\HH_{AB}$ and
$$
\mathfrak{D}\!\left(+\infty\right)\doteq\left\{\Lambda \in\B_{\rm sa}(\HH_{AB})\,\left|\, \langle\psi|\Lambda|\psi\rangle
        \leq S(\psi_A)\;\, \forall\psi\in\HH^1_{AB}\right.\right\}.
$$
The first expression here is the double Fenchel transform of the
concave function\break $\rho\mapsto S(\rho_A)$ on $\St(\HH_{AB})$. So, this expression follows, by the Fenchel-Moreau theorem, from the fact that
$E_F$ is the convex closure of that function \cite{EM}. The second  expression is derived from the first one by the simple arguments  (similar to the ones that were used in the below proof of Proposition \ref{VEP}.)\medskip

The following proposition shows that we may treat $E_F^\lambda$ as a "truncated"
double Fenchel transform of the function
$\rho \mapsto S(\rho_A)$. In this proposition we write $D(\Lambda)$ for the diameter of the spectrum of an operator $\Lambda\in\B_{\rm sa}(\HH_{AB})$:
\begin{equation}\label{D-def}
D(\Lambda)=\max \mathrm{Sp}(\Lambda)-\min \mathrm{Sp}(\Lambda).
\end{equation}

\begin{proposition}\label{VEP}
\emph{For an arbitrary  state $\rho\in\St(\HH_{AB})$ and  $\lambda>0$ the following expressions hold
\begin{equation}\label{VE}
E_F^\lambda(\rho)= \max_{\Lambda\in \B\left(\!\frac{1}{2\lambda}\!\right)} \inf_{\psi\in\HH^1_{AB}}\left\{S(\psi_A)+\Tr \Lambda(\rho-\psi)\right\}= \max_{\Lambda\in \mathfrak{D}\left(\!\frac{1}{\lambda}\!\right)}\Tr \Lambda\rho,
\end{equation}
where $\,\B(r)\doteq\left\{\Lambda \in\B_{\rm sa}(\HH_{AB})\,\left|\,\|\Lambda\|\leq r\right.\right\}$, $\,\HH^1_{AB}\doteq\{\psi \in\HH_{AB}\,|\,\|\psi\|=1\}\,$ and}\footnote{ $D(\Lambda)$ is defined in (\ref{D-def}).}
$$
\mathfrak{D}(d)\doteq\left\{\Lambda \in\B_{\rm sa}(\HH_{AB})\,\left|\, D(\Lambda)\leq d,\;\langle\psi|\Lambda|\psi\rangle
        \leq S(\psi_A)\;\, \forall\psi\in\HH^1_{AB}\right.\right\}.
$$
\end{proposition}

\begin{remark}\label{D-r}
It is easy to see that the condition $\,\langle\psi|\Lambda|\psi\rangle
        \leq S(\psi_A)\;\, \forall\psi\in\HH^1_{AB}\,$ used in the definition of the set $\mathfrak{D}(d)$ is equivalent to the following one
$$
\Tr\Lambda\sigma\leq E_F(\sigma)\quad \forall\sigma\in\St(\HH_{AB}).
$$
\end{remark}

\emph{Proof.} To prove the first equality in (\ref{VE}) note that the second expression in (\ref{EFA}) implies
$$
E_F^\lambda(\rho)=\inf_{\mu\in\P(\HH_{AB})}\max_{\Lambda\in\B\left(\!\frac{1}{2\lambda}\!\right)}G(\mu,\Lambda)=\max_{\Lambda\in\B\left(\!\frac{1}{2\lambda}\!\right)}\inf_{\mu\in\P(\HH_{AB})}G(\mu,\Lambda),
$$
where  $G(\mu,\Lambda)\doteq S_A(\mu)+\Tr \Lambda(\rho-\bar{\rho}(\mu))$. The first equality here is obvious, the second one follows
from Theorem 3.1 in \cite{Sim}, since the function
$G(\mu,\Lambda)$ is affine in both arguments and continuous in $\Lambda$ w.r.t. the $\sigma$-weak operator topology on $\B_{\rm sa}(\HH_{AB})$, while the ball $\B\!\left(\frac{1}{2\lambda}\right)$ is
compact in this topology.\smallskip

Thus, to prove (\ref{VE}) it suffices to show that
$$
\inf_{\mu\in\P(\HH_{AB})}G(\mu,\Lambda)=\inf_{\psi\in\HH^1_{AB}}H(\psi,\Lambda),\quad \textrm{where} \quad H(\psi,\Lambda)\doteq S(\psi_A)+\Tr \Lambda(\rho-\psi).
$$
This can be done by the way described in  the proof of the implication $\,\rm (i)\Rightarrow(iii)\,$ in Theorem 1 in \cite{H&Sh-5}.

The second equality in (\ref{VE}) is easily derived from the first one by noting  that
$\,H(\psi,\Lambda)=H(\psi,\Lambda+cI_{AB})\,$ for any $c\in\mathbb{R}$. Indeed, the last property implies
\begin{equation}\label{tpt}
\begin{aligned}
\max_{\Lambda\in \B(\frac{1}{2\lambda})}
\inf_{\psi\in\mathcal{H}^1_{AB}}
H(\Lambda,\psi)
&=
\max\left\{
    \operatorname{Tr}\Lambda\rho
    \,\middle|\,
    \Lambda\in \B_\lambda,\;
    g(\Lambda)=0
\right\},
\end{aligned}
\end{equation}
where
\[
    \B_\lambda
    =
    \left\{
        \Lambda\in\B_{\mathrm{sa}}(\mathcal{H}_{AB})
        \,\left|\,
        D(\Lambda)\leq \frac{1}{\lambda}
    \right.\right\}
\]
and
\[
    g(\Lambda)=\inf_{\psi\in\mathcal{H}^1_{AB}}\left\{S(\psi_A)-\langle\psi|\Lambda|\psi\rangle
    \right\}.
\]

The r.h.s. of (\ref{tpt}) is not greater than the last expression in (\ref{VE}),
since $g(\Lambda)=0$ implies that
\[
    \langle\psi|\Lambda|\psi\rangle
    \leq S(\psi_A)
    \qquad
    \forall\,\psi\in\HH^1_{AB}.
\]

Denote the r.h.s. of (\ref{tpt}) by $X$ and assume that
\[
    \sup\left\{
        \operatorname{Tr}\Lambda\rho
        \,\middle|\,
        \Lambda\in \B_\lambda,\;
        g(\Lambda)\geq 0
    \right\}>X.
\]
Then there is $\Lambda\in \B_\lambda$ such that $g(\Lambda)=\Delta>0$ and
$\operatorname{Tr}\Lambda\rho>X$. But in this case
\[
  \Lambda+\Delta I_{AB}\in\B_\lambda,\quad   g(\Lambda+\Delta I_{AB})=0\quad \textrm{and}\quad
    \operatorname{Tr}(\Lambda+\Delta I_{AB})\rho
    =
    \operatorname{Tr}\Lambda\rho+\Delta,
\]
contradicting our assumption.  $\Box$\smallskip

Variational expression  (\ref{VE}) can be used to obtain alternative proofs of basic properties of the function $E_F^\lambda$
established in Section 3.1.

For example, to prove the nonselective LOCC-monotonicity of function $E_F^\lambda$ for a given $\lambda>0$  take  arbitrary state $\rho\in\St(\HH_{AB})$ and
assume that $\Phi$ is a channel describing a given nonselective LOCC-operation in the system $AB$. Let
$\Lambda$ be an optimal operator for the state $\Phi(\rho)$ (w.r.t. expression  (\ref{VE})), i.e. such an operator in $\B_{\rm sa}(\HH_{AB})$ that
\begin{equation}\label{Lp}
E_F^\lambda(\Phi(\rho))=\Tr \Lambda\Phi(\rho),\quad \Tr \Lambda\sigma\leq E_F(\sigma)\quad \forall\sigma\in\St(\HH_{AB})\quad \textrm{and} \quad D(\Lambda)\leq \frac{1}{\lambda}.
\end{equation}
Consider the operator $\Phi^*(\Lambda)$, where $\Phi^*$ is the dual map to the channel $\Phi$ \cite{H-SCI,Wilde}. Then we have
\begin{equation}\label{Lp+}
\Tr\Phi^*(\Lambda)\sigma=\Tr\Lambda\Phi(\sigma)\leq E_F(\Phi(\sigma))\leq E_F(\sigma)\quad \forall\sigma\in\St(\HH_{AB}),
\end{equation}
where the first inequality holds by the second property in (\ref{Lp}) and the second one is due to the LOCC-monotonicity of the EoF.
Since the map $\Phi^*$ is unital and does not increase the operator norm, it is easy to prove that $D(\Phi^*(\Lambda))\leq \frac{1}{\lambda}$. This and (\ref{Lp+}) show, by Remark \ref{D-r},
that the operator $\Phi^*(\Lambda)$ belongs to that set $\mathfrak{D}(d)$. Hence,  the expression  (\ref{VE}) implies that
$$
E_F^\lambda(\rho)\geq \Tr\Phi^*(\Lambda)\rho=\Tr\Lambda\Phi(\rho)=E_F^\lambda(\Phi(\rho)).
$$

Note  that variational expressions (\ref{EV}) and (\ref{VE}) directly imply claim $C$ of Theorem \ref{P-1}.
Variational expressions  (\ref{VE}) can be also used to prove the Lipschitz continuity bound for the function $E_F^\lambda$ in claim $B$ of Theorem \ref{P-1}.

\subsection{Subadditivity of the function $E^{\lambda}_F(\rho)$ under tensor products and beyond}

Note first that the functional $E_F^\lambda$ inherits the subadditivity property from the EoF.\smallskip

\begin{proposition}\label{sd} \emph{ Let $A_1$, $A_2$, $B_1$ and $B_2$ be quantum systems of any dimensions. Let $\lambda>0$.
For arbitrary states $\rho$ and $\sigma$ of the systems $A_1B_1$ and  $A_2B_2$ the following inequality holds
\begin{equation}\label{eq:EF-lambda-subadditivity}
E_F^\lambda(\rho\otimes\sigma)
\leq
E_F^\lambda(\rho)+E_F^\lambda(\sigma),
\end{equation}
where $\rho\otimes\sigma$ is treated as a state of the bipartite system $(A_1A_2)(B_1B_2)$.}
\end{proposition}\smallskip

\emph{Proof.} Let $\rho'$ and $\sigma'$ be arbitrary states of the systems $A_1B_1$ and  $A_2B_2$. By the
subadditivity of the EoF,
\begin{equation}
E_F(\rho'\otimes\sigma')
\leq
E_F(\rho')+E_F(\sigma').
\label{eq:EF-subadditivity}
\end{equation}
On the other hand, the trace norm satisfies
\begin{equation}
\|\rho\otimes\sigma-\rho'\otimes\sigma'\|_1
\leq
\|\rho\otimes\sigma-\rho'\otimes\sigma\|_1
+
\|\rho'\otimes\sigma-\rho'\otimes\sigma'\|_1.
\end{equation}
Using $\,\rho\otimes\sigma-\rho'\otimes\sigma=(\rho-\rho')\otimes\sigma\,$
and
$\,\rho'\otimes\sigma-\rho'\otimes\sigma'=\rho'\otimes(\sigma-\sigma')$,
together with $\,\|\alpha\otimes \beta\|_1=\|\alpha\|_1\|\beta\|_1$, we obtain
\begin{equation*}
\|\rho\otimes\sigma-\rho'\otimes\sigma'\|_1
\leq
\|\rho-\rho'\|_1+\|\sigma-\sigma'\|_1.
\end{equation*}

Consequently,
\begin{equation*}
\begin{array}{c}
E_F(\rho'\otimes\sigma')
+\frac{1}{2\lambda}
\|\rho\otimes\sigma-\rho'\otimes\sigma'\|_1
\leq
E_F(\rho')+E_F(\sigma')
+\frac{1}{2\lambda}\|\rho-\rho'\|_1
+\frac{1}{2\lambda}\|\sigma-\sigma'\|_1
\\\\=
\left[
E_F(\rho')+\frac{1}{2\lambda}\|\rho-\rho'\|_1
\right]
+
\left[
E_F(\sigma')+\frac{1}{2\lambda}\|\sigma-\sigma'\|_1
\right].
\end{array}
\end{equation*}

Since $\rho'$ and $\sigma'$ are arbitrary, taking the infimum over
$\rho'$ and $\sigma'$ on the right-hand side gives
$$
E_F^\lambda(\rho\otimes\sigma)
\leq
\inf_{\rho'}
\left\{
E_F(\rho')+\frac{1}{2\lambda}\|\rho-\rho'\|_1
\right\}
+
\inf_{\sigma'}
\left\{
E_F(\sigma')+\frac{1}{2\lambda}\|\sigma-\sigma'\|_1
\right\}=E_F^\lambda(\rho)+E_F^\lambda(\sigma),
$$
which proves (\ref{eq:EF-lambda-subadditivity}). $\Box$\smallskip

There is also a useful lower bound on $E_F^\lambda(\rho\otimes\sigma)$. Since the partial trace is a local
quantum operation and $E_F^\lambda$ is monotone under such operations by Theorem \ref{P-1}F,
tracing out either factor gives
\begin{equation*}
E_F^\lambda(\rho)
\leq
E_F^\lambda(\rho\otimes\sigma)
\end{equation*}
and, similarly,
\begin{equation*}
E_F^\lambda(\sigma)
\leq
E_F^\lambda(\rho\otimes\sigma).
\end{equation*}
Therefore,
\begin{equation*}
\max\left\{
E_F^\lambda(\rho),E_F^\lambda(\sigma)
\right\}
\leq
E_F^\lambda(\rho\otimes\sigma)
\leq
E_F^\lambda(\rho)+E_F^\lambda(\sigma).
\end{equation*}

In particular, if $\sigma$ is separable, then
$E_F^\lambda(\sigma)=0$, and hence
\begin{equation*}
E_F^\lambda(\rho\otimes\sigma)
=E_F^\lambda(\rho).
\end{equation*}

Proposition \ref{sd} implies that
\begin{equation}
E_F^\lambda(\rho^{\otimes n})
\leq
n E_F^\lambda(\rho),
\qquad n\in\mathbb N.
\label{eq:tensor-power-subadditivity}
\end{equation}
So, Fekete's lemma allows us to define the regularized version of $E_F^\lambda$ by the expression
\begin{equation}
E_F^{\lambda,\infty}(\rho)\doteq\lim_{n\to\infty}
\frac{1}{n}E_F^\lambda(\rho^{\otimes n})=
\inf_{n\geq1}
\frac{1}{n}E_F^\lambda(\rho^{\otimes n}).
\label{eq:Fekete}
\end{equation}
Unfortunately, Theorem \ref{P-1}A implies that $E_F^{\lambda,\infty}(\rho)=0$ for any state $\rho$, so we cannot
assume that $E_F^{\lambda,\infty}$ tends to $E_F^{\infty}$ as $\lambda\to0^+$. \medskip

Note, finally,  that \emph{the function $E_F^\lambda$ is not additive}: there exist bipartite states $\rho$ and $\sigma$ such that
$$
E_F^\lambda(\rho\otimes\sigma)
<
E_F^\lambda(\rho)+E_F^\lambda(\sigma)
$$
for all $\lambda$ small enough.  This can be shown  easily by using
the nonadditivity of the EoF (proved by Hastings in \cite{Hast} using the Shor theorem \cite{Shor}) and claim C of Theorem \ref{P-1}.

\section{Upper bounds on $E_F(\rho)-E^{\lambda}_F(\rho)$}

By Theorem \ref{P-1}C the family $\left\{E_F^\lambda\right\}_{\lambda>0}$
of MY-approximations pointwise converges to the function
$E_F$ as $\lambda\to0^+$. In this section we show how to obtain
easily computable estimates for the rate of this convergence.
For this purpose we use the (lower) semicontinuity bounds for the function $E_F$
obtained in~\cite{LCB+}.\smallskip

The below results are based on the following lemma.
\newpage

\begin{lemma}\label{vil}
\emph{Let $A$ and $B$ be quantum systems of any dimensions.
Let $\rho$ be a state of the system $AB$.}

\medskip

A) \emph{If there is a function $B_T$ on $[0,1]$ such that
\[
    E_F(\rho)-E_F(\sigma)
    \leq
    B_T(\varepsilon)
    \qquad
\]
for all  $\,\sigma\in\St(\mathcal{H}_{AB})$ such that $\,\frac{1}{2}\|\rho-\sigma\|_1\leq\varepsilon\,$ then
\begin{equation}\label{st1}
   E_F(\rho)-E_F^\lambda(\rho)
    \leq
    \sup_{\varepsilon\in[0,1]}
    \left\{B_T(\varepsilon)-\frac{\varepsilon}{\lambda}\right\}.
    \end{equation}
If $\,B_T(\varepsilon)\to0\,$ as $\,\varepsilon\to0^+$, then the r.h.s.
of (\ref{st1}) tends to zero as $\,\lambda\to0^+$.}
\medskip

B) \emph{If there is a function $B_F$ on $[0,1]$ such that
\[
    E_F(\rho)-E_F(\sigma)\leq B_F(\delta)
    \qquad
\]
for all  $\,\sigma\in\St(\mathcal{H}_{AB})$ such that $\,\sqrt{1-F(\rho,\sigma)}\leq\delta ,$ where $\,F(\rho,\sigma)\doteq\|\sqrt{\rho}\sqrt{\sigma}\|_1^2\,$ is the fidelity between $\rho$ and $\sigma$, then
\begin{equation}\label{st2}
    E_F(\rho)-E_F^\lambda(\rho)
    \leq
    \sup_{\delta\in[0,1]}
    \left\{
        B_F(\delta)
        -\frac{1}{\lambda}
        \left(
            1-\sqrt{1-\delta^2}
        \right)
    \right\}.
\end{equation}
If $\,B_F(\delta)\to0\,$ as $\,\delta\to0^+$, then the r.h.s.
of (\ref{st2}) tends to zero as $\,\lambda\to0^+$.}
\end{lemma}

\bigskip

\textit{Proof.} A) By the definition of $E_F^\lambda$ we have
\[
    E_F(\rho)-E_F^\lambda(\rho)
    =
    \sup_{\sigma\in\St(\HH_{AB})}
    \left\{
        E_F(\rho)-E_F(\sigma)
        -\frac{1}{2\lambda}\|\rho-\sigma\|_1
    \right\}\leq\sup_{\varepsilon\in[0,1]}
    \left\{B_T(\varepsilon)
        -\frac{\varepsilon}{\lambda}
    \right\}.
\]
B) To prove claim B it suffices to use the above arguments along with  the inequality (cf. \cite{H-SCI,Wilde})
$$
\frac{1}{2}\|\rho-\sigma\|_1\geq 1-\sqrt{F(\rho,\sigma)}.
$$
$\Box$

\medskip

\begin{example}\label{new}
Assume that $\,\rho$ is a  state in $\St(\HH_{AB})$ such that (\ref{lsc}) holds for some finite $\lambda_\rho>0$.  By applying Lemma \ref{vil}A  we conclude that
\begin{equation*}
    E_F(\rho)-E_F^\lambda(\rho)
    \leq
    \max_{x\in[0,1]}   \left\{\left(\frac{1}{\lambda_\rho}-\frac{1}{\lambda}\right)x\right\}=
\begin{cases}
0 & \textup{if }\;\lambda\leq \lambda_\rho\\[2mm]
\frac{1}{\lambda_\rho}-\frac{1}{\lambda} &\textup{if }\; \lambda>\lambda_\rho
\end{cases}.
\end{equation*}
This bound agrees with claim D of Theorem \ref{P-1}.
\end{example}

By using claim B of Lemma \ref{vil} and the semicontinuity bounds for the function $E_F$
from  Proposition 9 in \cite{AOE} and Proposition 10 in \cite{LCB+} modified according to Remark  15 in \cite{LCB+} we obtain the following

\smallskip

\begin{proposition}\label{P-2} \emph{Let $A$ and $B$ be quantum systems of any dimensions.
Let $\rho$ be a state of the system $AB$. Let $\,h$ be the binary entropy and $\,h^\uparrow$ its non-decreasing
envelope defined in (\ref{h+}).}
\smallskip

A) \emph{If $\,2\leq\rank \rho_A = d <+\infty,$ then}
\begin{equation}\label{pt-1}
    E_F(\rho)-E_F^\lambda(\rho)
    \leq
    \max_{x\in(0,1-1/d\shs]}
    \left\{
        x\ln(d-1)
        +h(x)
        -\frac{1}{\lambda}
        \left(1-\sqrt{1-x^2}\right)
    \right\}.
\end{equation}\medskip

B) \emph{If $\,\Tr H\rho_A = E < +\infty\,$ for some positive operator $H$ on $\HH_A$ satisfying conditions (\ref{H-cond}) and (\ref{star}), then
\begin{equation}\label{pt-2}
    E_F(\rho)-E_F^\lambda(\rho)
    \leq
    \max_{x\in(0,1]}
    \left\{
        x F_H\!\left(\frac{E}{x}\right)
        + h^\uparrow(x)
        - \frac{1}{\lambda}
          \left(1-\sqrt{1-x^2}\right)
    \right\}.
\end{equation}
where $F_H$ is the function defined in (\ref{F-def})}.\medskip

\emph{The r.h.s. of (\ref{pt-1}) and (\ref{pt-2}) tends to zero as $\lambda\to0^+$.\footnote{The r.h.s. of (\ref{pt-2}) tends to zero as $\lambda\to0^+$ due to the equivalence of (\ref{H-cond}) and (\ref{H-cond-a}).}
}\end{proposition}
\medskip

\textbf{Note A:} Claim B of Proposition \ref{P-2} is applicable to any state $\rho$ in $\St(\HH_{AB})$ such that $\,S(\rho_X)<+\infty,$  where $X$ is either $A$ or $B$, since for any such state there exists positive operators $H$ on $\HH_X$ satisfying  conditions (\ref{H-cond}) and (\ref{star}) such that $\,\Tr H\rho_X<+\infty$  \cite[Proposition 4]{EC}.\medskip

\textbf{Note B:} Using claim A of Lemma \ref{vil} and the original (non-modified) semicontinuity bounds for the function $E_F$
from Proposition 9 in \cite{AOE} and Proposition 10 in \cite{LCB+} one can obtain analogues of the bounds (\ref{pt-1}) and (\ref{pt-2}), but numerical analysis shows that these analogues
are less accurate than the bounds (\ref{pt-1}) and (\ref{pt-2}). \medskip

Proposition \ref{P-2} imply the following upper bounds on the rate of uniform convergence of
the MY-approximation  $E_F^\lambda$  to the EoF as $\lambda\to0^+$ on the set of states with bounded rank/energy  of one of the marginal states.\smallskip

\begin{corollary}\label{C-1} \emph{The family $\{E^\lambda_F\}_{\lambda\geq 0}$
of MY-approximations uniformly converges to the function $E_F$ on
the sets
\[
    \C^A_d
    =
    \{\rho\in\St(\mathcal{H}_{AB})\, |\,
    \rank\rho_A \leq d<+\infty\}
\]
and
\[
    \C^A_{H,E}
    =
    \{\rho\in\St(\mathcal{H}_{AB})\, |\,
    \Tr H\rho_A \leq E\},
\]
where $H$ is any positive operator on $\HH_A$ satisfying  conditions (\ref{H-cond}) and (\ref{star}). Moreover,
\begin{equation}\label{ub1}
    \sup_{\rho\in \C^A_d}
    \left|E_F(\rho)-E^\lambda_F(\rho)\right|
    \leq  \sup_{x\in(0,1-1/d]}
    \left\{
        x\ln(d-1)
        +h(x)
        -\frac{1}{\lambda}
        \left(1-\sqrt{1-x^2}\right),
    \right\}
\end{equation}
where $h$ is the binary entropy, and
\begin{equation}\label{ub2}
    \sup_{\rho\in\C^A_{H,E}}
    \left|E_F(\rho)-E^\lambda_F(\rho)\right|
    \leq  \sup_{x\in(0,1]}
    \left\{
        x F_H\!\left(\frac{E}{x}\right)
        + h^\uparrow(x)
        - \frac{1}{\lambda}
          \left(1-\sqrt{1-x^2}\right)
    \right\},
\end{equation}
where $F_H$ is the function defined in (\ref{F-def}) and $\,h^\uparrow$ is the function defined in (\ref{h+}).}
\smallskip

\emph{The r.h.s. of (\ref{ub1}) and (\ref{ub2}) tends to zero as $\lambda\to0^+$.}
\end{corollary}
\medskip\medskip

\begin{example}\label{1-exam}  It seems reasonable to characterize the accuracy of the uniform approximation of the EoF by the function $E^\lambda_F$ on the set $\C^A_{d}$ by the relative error
\begin{equation*}
{\rm re}_{d}(\lambda)\doteq\frac{\sup_{\rho\in \C^A_d}\left|E_F(\rho)-E^\lambda_F(\rho)\right| }{\sup_{\rho\in \C^A_d}E_F(\rho)}=\frac{\sup_{\rho\in \C^A_d}\left|E_F(\rho)-E^\lambda_F(\rho)\right| }{\ln d},
\end{equation*}
where we assumed that $\,\dim\HH_B\geq d$.  So, the accuracy of the bound (\ref{ub1}) is characterized by
\begin{equation}\label{rb1}
\widehat{\rm re}_{d}(\lambda)\doteq\frac{\textrm{the r.h.s. of }(\ref{ub1})}{\ln d}.
\end{equation}

The result of numerical calculations of $\,\widehat{\rm re}_{d}(\lambda)\,$ as a function of $\lambda$ for different values of $d$ are shown on Fig.1. We see that
the relative error bound slightly depends on $d$.
\end{example}\smallskip

\begin{example}\label{2-exam}
It seems reasonable to characterize the accuracy of the uniform approximation of the EoF by the function $E^\lambda_F$ on the set $\C^A_{H,E}$ by the relative error
\begin{equation*}
{\rm re}_{H,E}(\lambda)\doteq\frac{\sup_{\rho\in \C^A_{H,E}}\left|E_F(\rho)-E^\lambda_F(\rho)\right| }{\sup_{\rho\in \C^A_{H,E}}E_F(\rho)}=\frac{\sup_{\rho\in \C^A_{H,E}}\left|E_F(\rho)-E^\lambda_F(\rho)\right| }{F_H(E)},
\end{equation*}
where  the function $F_H$ defined in (\ref{F-def}) appears because we assumed that $\,\dim\HH_B=\infty$. So, the accuracy of the bound (\ref{ub2}) is characterized by
\begin{equation}\label{rb2}
\widehat{\rm re}_{H,E}(\lambda)\doteq\frac{\textrm{the r.h.s. of }(\ref{ub2})}{F_H(E)}.
\end{equation}

To obtain concrete estimates assume that $\,H=N\doteq a^\dag a\,$ is the number operator of a quantum oscillator  (see \cite[Section 12]{H-SCI}). In this case
$F_{N}(E)=g(E)$, where $g(x)$ is the function defined (\ref{g-def}). The result of numerical calculations of $\,\widehat{\rm re}_{N,E}(\lambda)\,$ as a function of $\lambda$ for different values of $E$ are shown on Fig.2. We see that the relative error bound is maximal for $E=1$ and slightly depends on $E$ for large $E$.

\begin{figure}[t]

\centering
\begin{center}

\includegraphics[scale=0.4, bb=400  450 500 550]{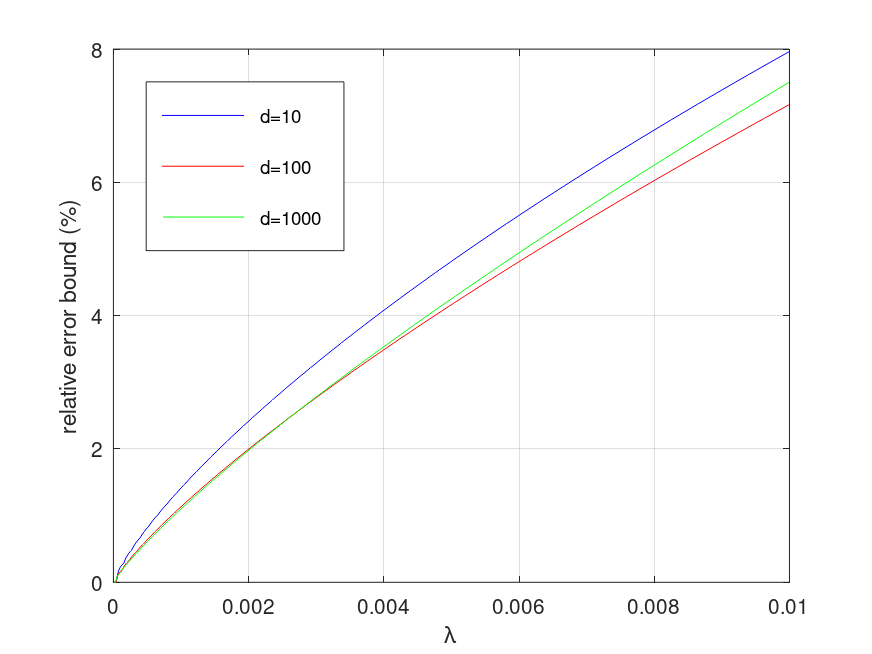}

\vspace{180pt}
\caption{The percentage value of $\,\widehat{\rm re}_{d}(\lambda)\,$ as a function of $\lambda$ for $d=10,100,1000$.}
\end{center}
\label{Fig1}
\end{figure}

\begin{figure}[t]

\centering
\begin{center}

\includegraphics[scale=0.4, bb=400 450 500 550]{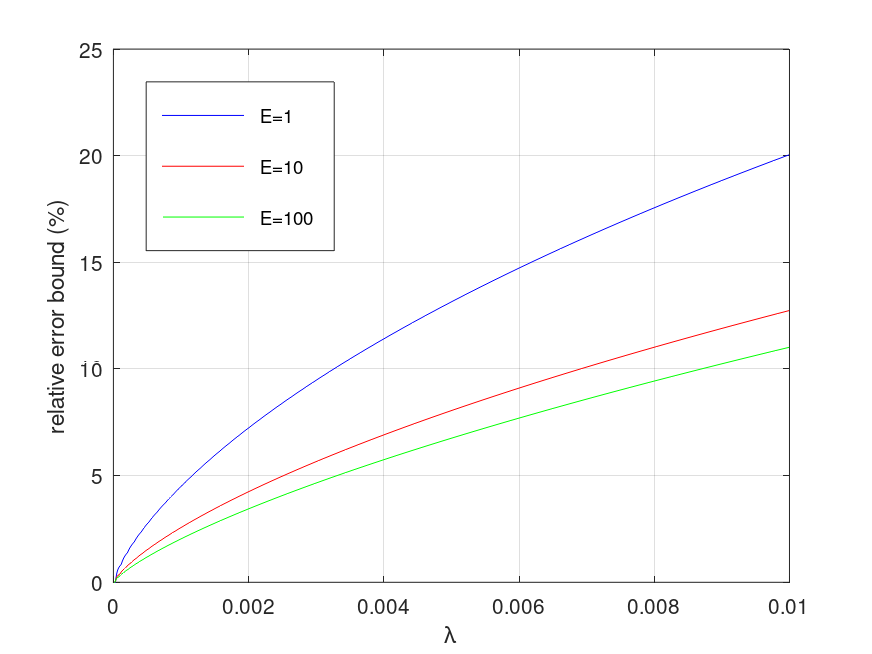}

\vspace{180pt}
\caption{The percentage value of $\,\widehat{\rm re}_{N,E}(\lambda)\,$ as a function of $\lambda$ for $E=1,10,100$.}
\end{center}

\label{Fig2}
\end{figure}
\end{example}

\section{The cases when $\,E_F(\rho)=E^\lambda_F(\rho)\,$ for small $\lambda>0$}

According to Theorem  \ref{P-1}D,  the equality $\,E_F^\lambda(\rho)=E_F(\rho)\,$ holds for given $\rho$ and $\lambda=\lambda_\rho>0$ (and hence for all $\,\lambda\in (0,\lambda_\rho]\,$) if and only if
\begin{equation}\label{lsc++}
E_F(\rho)-E_F(\sigma)\leq \frac{1}{2\lambda_\rho}\|\rho-\sigma\|_1\quad \forall\sigma\in\St(\mathcal{H}_{AB}).
\end{equation}
By Theorem 1 in \cite{H&Sh-5} this is equivalent to
the existence of an operator $\Lambda_\rho$ with the spectral diameter  $\,D(\Lambda_\rho)\leq\frac{1}{\lambda_\rho}\,$ such that,
\[
   E_F(\rho)
   =\operatorname{Tr}\Lambda_\rho\rho
   \quad\textrm{ and }\quad
   \langle\psi|\Lambda_\rho|\psi\rangle
   \leq S(\psi_A)
   \quad
   \forall\,\psi\in\mathcal{H}^1_{AB},
 \]
where $\mathcal{H}^1_{AB}$ is the unit sphere in $\mathcal{H}_{AB}$. The last property means
that  $\,\ell_\rho(\sigma)=\Tr\Lambda_\rho \sigma\,$  is an affine  supporting functional for the EoF
at the state $\rho$, i.e.
$$
\ell_\rho(\rho)=E_F(\rho)\quad\textrm{ and }\quad\ell_\rho(\sigma)\leq E_F(\sigma)\;\;\textrm{ for any state }\,\sigma\,\textrm{ in }\,\St(\HH_{AB}).
$$
It is shown in \cite{H&Sh-5}  that in every bipartite quantum system $AB$ there exist states
for which the above equivalent properties do not hold. So,  in every bipartite  quantum system there exist states $\rho$ for which
$$
E^\lambda_F(\rho)<E_F(\rho)\quad\textrm{ for all }\,\lambda>0.
$$
So, the convergence of $E^\lambda_F(\rho)$ to $E_F(\rho)$ as $\lambda\to0^+$ at these states is asymptotical.\smallskip

At the same time, several sufficient conditions for
the validity of (\ref{lsc++}) for a given state $\rho$ are found in \cite{H&Sh-5} and \cite{H&Sh-6}. By Theorem  \ref{P-1}D, these
conditions also imply
\begin{equation}\label{equ}
E_F^\lambda(\rho)=E_F(\rho)\quad  \forall \lambda\in (0,\lambda_\rho]
\end{equation}
for some $\lambda_\rho>0$ (depending on $\rho$).\smallskip

Below we briefly describe the cases when (\ref{equ}) holds. In
all the cases we assume that $A$ and $B$ are quantum systems of any dimensions.\smallskip

\textbf{Case A.}  Let $\rho$ be a pure state of the system $AB$ such that the marginal states
$\rho_A$ and $\rho_B$ has the spectrum
\begin{equation}\label{ud}
\Bigl(\underbrace{\frac{1}{d},\frac{1}{d},...,\frac{1}{d}}_{d\textrm{ items}},0,0...\Bigr),\quad d\geq 2.
\end{equation}

By Proposition 2 in \cite{H&Sh-6} the semicontinuity bound  (\ref{lsc++}) holds for all $\,\lambda\in (0,\frac{1}{K_d}]$, where
\begin{equation*}
K_d\doteq
\begin{cases}
2 & \textup{if }\; d=2,\\[1mm]
\dfrac{d}{d-2}\ln(d-1) & \textup{if }\;d>2,
\end{cases}
\end{equation*}
is the optimal constant in the log-Sobolev inequality for the uniform measure on the complete graph  with $d$ vertices \cite{LS-1,LS-2,H&Sh-6}.
It follows that (\ref{equ}) holds with $\lambda_\rho=\frac{1}{K_d}$ in this case.

It is somewhat surprising that, at the moment, we cannot prove or disprove the validity of (\ref{equ}) with some $\,\lambda_\rho>0\,$
even for all pure states of the simplest 2-qubit system (see the end of the Introduction in \cite{H&Sh-6}).\smallskip

\textbf{Case B.}  Let $\rho$ be a finite rank state of the system $AB$ such that
$$
2\leq d\doteq\min\{\rank \rho_A,\rank\rho_B\}<+\infty\quad\textrm{ and }\quad\supp\rho=\supp\rho_A\otimes\supp\rho_B.
$$
By Proposition 7A in \cite{H&Sh-5} the semicontinuity bound  (\ref{lsc++}) holds for all $\,\lambda\in (0,\frac{1}{C^r_\rho}]$, where
$$
C^r_\rho\doteq\frac{\log d-E_F(\rho)}{\lambda_{\min}^{\rho}},\quad  \lambda_{\min}^{\rho}\textup{ is the minimal positive eigenvalue of }\rho.
$$
It follows that (\ref{equ}) holds with $\lambda_\rho=\frac{1}{C^r_\rho}$ in this case.\smallskip

\textbf{Case C.}  Let $\rho$ be a finite rank state of the system $AB$ such that
$$
0<\rm{me}(\rho)\doteq\min\{S(\rho_A),S(\rho_B)\}<+\infty\quad\textrm{ and }\quad\supp\rho=\supp\rho_A\otimes\supp\rho_B.
$$
By Proposition 7B in \cite{H&Sh-5} the semicontinuity bound  (\ref{lsc++}) holds for all $\,\lambda\in (0,\frac{1}{C^s_\rho}]$, where
$$
C^s_\rho\doteq\left[
\frac{\rm{me}(\rho)}
{(\lambda_{\min}^{\rho})^2}
-
\frac{E_F(\rho)}
{\lambda_{\min}^{\rho}}
\right],\quad  \lambda_{\min}^{\rho}\textup{ is the minimal positive eigenvalue of }\rho.
$$
It follows that (\ref{equ}) holds with $\lambda_\rho=\frac{1}{C^s_\rho}$ in this case.

\smallskip

\textbf{Note:} The condition  $\,\supp\rho=\supp\rho_A\otimes\supp\rho_B\,$ in cases
$B$ and $C$ is essential. This can be shown by the example of a state of the 2-qubit system at which the EoF has no
 affine  supporting functional presented in \cite{H&Sh-5}.
\bigskip

Using the bound in Example \ref{new} in Section 4 one can obtain upper bounds on the
difference $\,E_F(\rho)-E^{\lambda}_F(\rho)\,$  in  all the above cases A,B and C for $\lambda>\lambda_\rho$, which are more accurate than
the universal bounds given by Proposition \ref{P-2}.

\section{Summary of the results}

In the article we have defined and analyzed  the family $\{E^{\lambda}_F\}_{\lambda>0}$ of functions on the set $\St(\HH_{AB})$ (called Moreau-Yosida approximation of the EoF) such
$$
E^{\lambda}_F(\rho)\nearrow E_F(\rho)\leq +\infty\quad\textrm{ as }\; \lambda\searrow0^+\;\textrm{ for all }\;\rho\in\St(\HH_{AB}).
$$
Depending on the state $\rho$, there are two types of the convergence of $E^{\lambda}_F(\rho)$ to $E_F(\rho)$:
\begin{itemize}
  \item there exists $\lambda_{\rho}>0$ such that $E^{\lambda}_F(\rho)=E_F(\rho)$ for all $\lambda\in(0,\lambda_{\rho}]$ (this holds if and only if the equivalent properties
  in Theorem 1 in \cite{H&Sh-5} are valid for the state $\rho$);
  \item $E^{\lambda}_F(\rho)$ tends to $E_F(\rho)$ asymptotically as $\,\lambda\to0^+$.
\end{itemize}
In every nontrivial bipartite quantum system (in particular, in the simplest 2-qubit system) there exist states $\rho$ for which each of these cases takes place. \medskip

We have shown that for every $\lambda>0$
\begin{itemize}
  \item the function $E^{\lambda}_F$ is convex and  Lipschitz continuous on $\St(\HH_{AB})$ (with the Lipschitz constant not exceeding $\frac{1}{\lambda}$),
  \item the  equality  $\,E^\lambda_F(\rho)=0\,$  holds if and only if $\rho$ is a separable  state in $\St(\HH_{AB})$,
  \item the function $E^{\lambda}_F$  does not increase under nonselective
LOCC-operations \cite{{Vidal,P&V}},
  \item the function $E^{\lambda}_F$ is subadditive under tensor products,
  \item the definition of the function $E^\lambda_F$ is invariant with respect to the embeddings $\,\HH_A\subset\HH_{A'}$ and $\,\HH_B\subset\HH_{B'}$,
  \item the function $E^{\lambda}_F$ can be defined via the optimization over discrete ensembles of pure states only,
  \item the function $E^{\lambda}_F$ can be defined by the variational expression  which can be interpreted as a "truncation" of the well-known variational expression for the EoF.
  \end{itemize}
\medskip

To a large extent, the potential applicability of the constructed MY-approximations of the EoF  is based on the fact that
for all states except those for which both marginal entropies are infinite  there are easily calculated estimates of the accuracy of these approximations, i.e. upper bounds
on the difference
$$
E_F(\rho)-E^{\lambda}_F(\rho)
$$
tending to zero as  $\,\lambda\to0^+$ and depending on simple parameters of the state $\rho$ (rank/energy of one of the marginal states of $\rho$).
In a sense, these bounds reduce the problem of calculation of the EoF  with a given accuracy to the problem of
calculation of the function $E_F^{\lambda}$ for sufficiently small $\lambda>0$\bigskip

\section{Open question: can the function $E^{\lambda}_F$  increase under selective
LOCC-operations?}

\centerline{(this section is written with the help of ChatGPT-5.6)}
\bigskip

It is proved in Section 3 (Corollary 1) that the function $E^{\lambda}_F$  does not increase under nonselective
LOCC-operations  for every $\lambda>0$.  Unfortunately, the (quite simple) arguments  used to prove this property are not generalized to
selective LOCC-operations. At the same time, all the attempts of ChatGPT-5.6 to find a counterexample were unsuccessful. So, \emph{we may conjecture, at the moment,
that the function $E^{\lambda}_F$ does not increase under selective
LOCC-operations} as well. ChatGPT-5.6 tried (so far unsuccessfully) to use different ways to prove this conjecture. One of them was to apply the variational expression
(\ref{VE}) for the function $E^{\lambda}_F$  (generalizing the proof of the nonselective LOCC-monotonicity presented at the end of Section 3.2).

We briefly describe one of these ways (which is the most promising by the opinion of ChatGPT-5.6).

By Theorem 3 in \cite{MH}  to prove selective LOCC-monotonicity of the function $E^{\lambda}_F$ (for a given $\,\lambda>0\,$) it suffices\footnote{The convexity of $E^{\lambda}_F$ follows from Theorem \ref{P-1}B, the invariance of $E^{\lambda}_F$  under local unitaries and under adding local ancillary systems follows from Theorem \ref{P-1}F applied to the corresponding local channels and their inverses. Hence, by Theorem 3 in [10], it remains to prove (\ref{FL}).} to show that
\begin{equation}\label{FL}
E_{F}^{\lambda}
\left(
\sum_i p_i\,\rho_i\otimes |i\rangle\langle i|_F
\right)
=
\sum_i p_i E_{F}^{\lambda}(\rho_i)
\end{equation}
for any ensemble $\{p_i,\rho_i\}_i$ of states in  $\St(\mathcal H_{AB})$, where \(\{|i\rangle_F\}_i\) is an orthonormal basis in a local flag system $F=A',B'$.

The inequality $"\le"$ in  (\ref{FL}) follows immediately from  the convexity of the function $E_{F}^{\lambda}$,
because it is easy to see that $\,E_{F}^{\lambda}(\rho_i\otimes |i\rangle\langle i|_F)=E_{F}^{\lambda}(\rho_i)\,$ (this can be done using the definition of $E_{F}^{\lambda}$ and the non-increasing
of $E_{F}^{\lambda}$ under local channels).

Thus, the nontrivial part is the reverse
inequality
\begin{equation*}
E_{F}^{\lambda}(\omega)
\ge
\sum_i p_i E_{F}^{\lambda}(\rho_i),\quad \omega=\sum_i p_i
\rho_i\otimes |i\rangle\langle i|_F.
\end{equation*}

A sufficient condition for proving (\ref{FL})  would be to show
that the infimum in the definition
$$
E_{F}^{\lambda}(\omega)
=
\inf_{\sigma}
\left\{
E_{F}(\sigma)
+
\frac{1}{2\lambda}\|\omega-\sigma\|_1
\right\}
$$
can be restricted to states of the form
$$
\sigma
=
\sum_i p_i\,\sigma_i\otimes |i\rangle\langle i|_F
$$
with the same  probabilities $p_i$. For such states one has
$$
E_{F}(\sigma)
=
\sum_i p_i E_{F}(\sigma_i)
$$
and
$$
\|\omega-\sigma\|_1
=
\sum_i p_i\|\rho_i-\sigma_i\|_1.
$$

Consequently,
$$
E_{F}^{\lambda}(\omega)
=
\inf_{\{\sigma_i\}}
\sum_i p_i
\left[
E_{F}(\sigma_i)
+
\frac{1}{2\lambda}
\|\rho_i-\sigma_i\|_1
\right]=\sum_i p_i E_{F}^{\lambda}(\rho_i).
$$
Therefore, the essential missing step is to prove that an optimizer,
or at least a minimizing sequence, may be chosen to preserve the local
classical flag and its probability distribution.
\smallskip

ChatGPT-5.6 made a detail investigation of the validity of equality (\ref{FL}) for the simplest 2-qubit system and proved that it holds
for several classes of the states  $\,\omega=\sum_i p_i
\rho_i\otimes |i\rangle\langle i|_F\,$ (and showed that for these states the optimizer has the form described before!),  but a general rigorous proof was not presented. At the same time, no counterexamples have been found.
\smallskip

Confirmation of the  hypothesis of selective LOCC-monotonicity of the function $E^{\lambda}_F$ would imply
that this function is a true entanglement monotone (in terms of \cite{Vidal,P&V,4H}).\footnote{I would be grateful for any comments concerning this question.}

\bigskip \bigskip

I am grateful to A.S.Holevo for the collaboration that motivated this research.

\end{document}